\documentclass[preprint, authoryear, 5p]{elsarticle}
\journal{Insurance: Mathematics and Economics}
\makeatletter
\def\ps@pprintTitle{%
  \let\@oddhead\@empty
  \let\@evenhead\@empty
  \def\@oddfoot{\reset@font\hfil\thepage\hfil}
  \let\@evenfoot\@oddfoot
}
\makeatother

\usepackage{lineno,hyperref}
\usepackage{amsmath, amsthm, amssymb}
\usepackage{mathtools}
\usepackage{bbm}
\usepackage{csquotes}
\usepackage{pgfplots}
\usepackage{pgfplotsthemetol}
\pgfplotsset{compat=1.11}
\usepgfplotslibrary{fillbetween}
\usepgfplotslibrary{groupplots}
\usepackage{tikz}
\usepackage{tikz-cd}
\usetikzlibrary{arrows, positioning}
\usepackage{subcaption}
\usepackage{enumitem}
\usepackage{bm}
\usepackage{float}
\usepackage{multicol}
\usepackage{multirow}

\newtheorem{theorem}{Theorem}[section]
\newtheorem{definition}[theorem]{Definition}

\newtheorem{proposition}[theorem]{Proposition}
\newtheorem{remark}[theorem]{Remark}
\newtheorem*{example}{Example}

\biboptions{authoryear}

\providecommand{\Pb}{\mathbb{P}}				

\providecommand{\rbraces}[1]{\left( #1 \right)} 		
\providecommand{\cbraces}[1]{\left[ #1 \right]}			
\providecommand{\curly}[1]{\left\{ #1 \right\}} 		

\providecommand{\rInt}[4]{\int\limits_{#2}^{#3} #1 \ #4}				
\providecommand{\cInt}[4]{\rInt{#1}{#2}{#3}{\mathrm{d}#4}}				

\providecommand{\TD}[1]{\Lambda(#1)}		
\providecommand{\TDL}[1]{\TD{#1}}						
\providecommand{\TDI}[2]{\Lambda_{#1}(#2)}						

\providecommand{\TDC}{\lambda}								

\providecommand{\tdo}{\leq_{td}}

\begin{document}

\begin{frontmatter}

\title{Comparing and quantifying tail dependence}
\author[tud]{Karl Friedrich Siburg}%
\author[tud]{Christopher Strothmann\fnref{scholarship}}%
\fntext[scholarship]{Supported by the German Academic Scholarship Foundation.}%
\author[ul]{Gregor Weiß\corref{correspondingauthor}}%
\cortext[correspondingauthor]{Corresponding author.}%
\address[tud]{Faculty of Mathematics, TU Dortmund University, Germany}%
\address[ul]{Faculty of Economics, Leipzig University, Germany}

\begin{abstract}
We introduce a new stochastic order for the tail dependence between random variables. We then study different measures of tail dependence which are monotone in the proposed order, thereby extending various known tail dependence coefficients from the literature.
We apply our concepts in an empirical study where we investigate the tail dependence for different pairs of S{\&}P 500 stocks and indices, and illustrate the advantage of our measures of tail dependence over the classical tail dependence coefficient.
\end{abstract}

\end{frontmatter}

\section{Introduction and results}

Estimating the level of tail dependence in a given data set is an integral part of modeling the distribution of financial and actuarial time series. 
Accounting for  tail dependence promises to enable the statistician to identify the risk of joint extreme comovements in multivariate data sets which would not be detected if one were to rely on linear correlation alone. 
Not surprisingly, coefficients of tail dependence that capture the asymptotic probability of such extreme comovements have been frequently used in the empirical finance literature.

In this paper, we propose a stochastic order on the set of tail dependence functions of bivariate random vectors, discuss measures that are monotone in said order, and employ all measures of tail dependence in an extensive empirical study. For the sake of better readability, we present our findings only for the bivariate case of two random variables. We point out, however, that all the results  in this paper are also valid in the multivariate case.

The interaction between random variables takes many forms and dependence modelling has introduced various concepts to describe the wide array of possible relationships.
Applications include the returns of two stocks in finance or the water levels of two nearby rivers.
These dependence concepts can treat the dependence as a whole, e.g. the correlation coefficient or concordance measures, or focus on specific aspects such as extremal events.
A classical measure of extremal dependence is the so-called (lower) \emph{tail dependence coefficient}, describing the comovement of $X$ and $Y$ in the tails of their joint distribution. 
Given two random variables $X$ and $Y$ with univariate distribution functions $F_X$ and $F_Y$, respectively, the (lower) \emph{tail dependence coefficient} is defined as
$$ \TDC := \lim\limits_{s \searrow 0} \frac{\mathbb{P} \rbraces{X \leq F_X^{-1}(s) \mid Y \leq F_Y^{-1}(s)}}{s} ~. $$
Considering the previous two applications, the tail dependence coefficient describes the probability of jointly occurring losses or severe drought events. 
For an extensive introduction to tail dependence and its properties, see for example \cite{Joe.2015}. 

The number $\TDC$, however, captures the extremal behaviour only along the diagonal and ignores all additional information concerning the asymmetric extremal behaviour of $X$ and $Y$. 
This is contrary to many empirical studies in recent years, which assert a, in general, distinctly asymmetric dependence between financial assets.
Therefore, in this paper, we will study the much more general (lower) \emph{tail dependence function}
\begin{equation} \label{eqn:def_tdf}
	\TDL{x,y} := \lim\limits_{s \searrow 0} \frac{\mathbb{P} \rbraces{X \leq F_X^{-1}(sx) \mid Y \leq F_Y^{-1}(sy)}}{s}
\end{equation}
with $x,y\geq 0$.
In particular, the tail dependence coefficient is just one single value of the tail dependence function $\TDL{x,y}$, namely $\TDC = \TDL{1,1}$.
We will see that the tail dependence function is a much finer quantity than the classical tail dependence coefficient and generates a variety of tail dependence measures that can distinguish between different extremal behaviours undetected by the tail dependence coefficient. 

Crucial to any systematic approach to measures of tail dependence is the existence of an underlying tail dependence order, since measures alone are not sufficient for a consistent comparison of the degree of tail dependence between two pairs of random variables. Let us explain this counterintuitive fact in the more familiar framework of dependence modelling where Spearman's $\rho$ and Kendall's $\tau$ are probably the two best-established measures of concordance. It is, however, \emph{not} possible to use their values for a consistent comparison of the degree of concordance between two pairs of random variables---indeed, there are $(X_1,Y_1)$ and $(X_2,Y_2)$ satisfying $\rho(X_1,Y_1) < \rho(X_2,Y_2)$ and $\tau(X_1,Y_1) > \tau(X_2,Y_2)$; see \citet[][Ex.~2.3.11]{Strothmann.2021}. This inherent inconsistency can only be resolved by the choice of an underlying order with respect to which both measures $\rho$ and $\tau$ are monotone.

Empirically, we find the following main results. First, index/stock combinations from the S\&P 500 exhibit both symmetric and asymmetric tail dependence functions. This underlines the need to take into account more than just the coefficient of tail dependence when assessing the random vector's tail behavior. Second, we find that different measures that are monotone in our proposed stochastic order of tail dependence produce different assessments of the strength of tail dependence. Finally, we highlight the significant range some of the monotone measures can take in value for the same coefficient of tail dependence.

Our paper is related to several empirical and theoretical studies across the risk management, actuarial science, and finance literature. In actuarial science, for example, \citet{EckertGatzertEJOR} model extreme losses using copulas to account for tail dependence (``ruin probability''). Coefficients of (usually lower) tail dependence have also been used in asset pricing by \citet{Rod,oki08,fack11,Meine,RuenziWeigert,irresbergerweiss} to study the risk premia for crash events. Finally, at the macro level, \citet{deJonghe} and \citet{OhPatton:2013-2,OhPatton:2013-1} employ measures of tail dependence to estimate the probability of a crash in the financial system. None of these applied studies, however, reflect on or discuss their sole use of the coefficients of tail dependence for measuring the degree of tail dependence. 

In a related strand of the literature, several studies have proposed alternative measures for total tail dependence \citep[see][]{Total}, heavy tailedness \citep[see][]{Ji}, tail negative dependence \citep[see][]{Negative}, quadrant and intermediate tail dependence \citep[see][]{Quadrant}. We point out, however, that these studies focus solely on specific measures of tail dependence and do not incorporate an underlying stochastic ordering.
Other approaches to extremal orderings have been proposed in the literature (see, e.g., \citet{Hua.2012b} and \citet{Li.2013} and the references therein) which are, however, more restrictive than the tail dependence order proposed in the present work. Finally, \citet{SIBURG2016241} propose an order and measures of the \textit{asymmetry} of a copula, but do not study the ordering of copulas with respect to their tail dependence.

The rest of the paper is organized as follows. 
Section~\ref{section:order} introduces our ordering of tail dependence and its key properties. 
A variety of new measures of tail dependence, all of them monotone in that ordering, are presented in Section~\ref{section:measures}, together with various examples comparing them to the classical tail dependence coefficient. The final Section~\ref{section:empirical} concludes our paper with an empirical study illustrating the various concepts for real financial data.

\section{Tail dependence order} \label{section:order}

Suppose you are given two pairs of random variables $(X_1,Y_1)$ and $(X_2,Y_2)$---how do you determine which of them is more tail dependent? 
One way to do this would be to order them by some measure, e.g., to evaluate the tail dependence coefficient on these two pairs and call that pair more tail dependent which has the larger tail dependence coefficient. 
Reasonable as this may seem, Figure~\ref{fig:tdc_drawbacks} illustrates the drawbacks of such an approach. 
While the pairs $(X_1, Y_1)$ and $(X_2, Y_2)$ possess the same tail dependence coefficient, their overall tail behaviour varies drastically. 
This difficulty leads to the construction and investigation of a more comprehensive stochastic order which allows for a meaningful comparison of $(X_1, Y_1)$ and $(X_2, Y_2)$. 
Building upon the chosen stochastic order, one would then like to construct stochastic measures which respect that order, i.e., are monotone with respect to that order. 
An extensive treatment of a variety of stochastic orders and their relations can be found in \cite{Shaked.2007}.

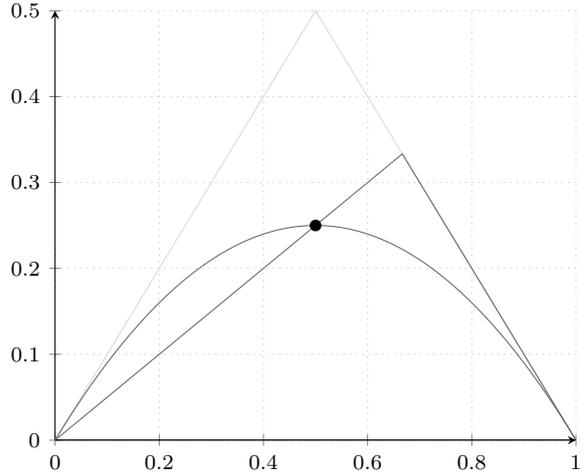
\begin{figure}
\begin{tikzpicture}[ declare function={ 
    			tdf(\t) 	= min(\t, 1-\t); 
				tdf1(\t)	= \t*(1-\t);
				tdf2(\t) 	= min(\t/2, 1-\t);
			}]
        \begin{axis}[mlineplot]
            \addplot[name path=f, domain=0:1, samples=151, color=gray!30]{tdf(\x)};
            \addplot[name path=g, domain=0:1, samples=151, color=black!60]{tdf1(\x)};
            \addplot[name path=h, domain=0:1, samples=151, color=black!60]{tdf2(\x)};
            \addplot[mark=*] coordinates {(1/2, {tdf1(0.5)})};
			\path[draw,name path = xAxis] (axis cs:0,0) -- (axis cs:1,0);
        \end{axis}
\end{tikzpicture}   
\caption{Plot depicting the tail dependence functions $\Lambda_1(s) = s(1-s)$ and $\Lambda_2(s) = \min\rbraces{s/3, 1-s}$ and their respective tail dependence coefficients $\lambda_1 = \lambda_2 = 1/4$.}
\label{fig:tdc_drawbacks}
\end{figure}

Since our aim is to investigate the tail behaviour of a pair of random variables, our proposed tail dependence order is based on the respective tail dependence functions $\Lambda_{(X_i,Y_i)}$ where we abbreviate $$ \TDI{i}{x,y} := \Lambda_{(X_i,Y_i)}(x,y) $$ for $i=1,2$. 
In the following, we will always assume that all pairs of random variables possess continuous univariate marginal distributions\footnote{See \cite{Feidt.2009} for a discussion of tail dependence in the presence of discontinuous univariate marginal distributions.} and admit a tail dependence function; we refer to \cite{Joe.2015} for more details.

\begin{definition}
We say that a pair of random variables $(X_1,Y_1)$ is less tail dependent than another pair $(X_2,Y_2)$, written $(X_1,Y_1) \tdo (X_2,Y_2)$, if and only if $$\TDI{1}{x,y} \leq \TDI{2}{x,y} $$ holds for all $x,y \geq 0$.
\end{definition}

We point out that $\tdo$ is only a preorder, i.e.\ a reflexive and transitive relation. It is neither antisymmetric (since all copulas with vanishing tail dependence function are comparable) nor total (see the example in Figure~\ref{fig:tdc_drawbacks}). Nevertheless, following common practice, we will call $\tdo$ the \emph{tail dependence order}.

The tail dependence function $\TDL{x, y}$ of a pair of random variables $(X,Y)$ is, by definition, a function on $[0,\infty) \times [0,\infty)$. 
It has, however, additional properties which we will need later that are collected in the following proposition; for the proof of these properties we refer to \cite{Gudendorf.2010}.

\begin{proposition} \label{prop:tdf_properties}
A function $\Lambda: [0, \infty) \times [0, \infty) \rightarrow [0, \infty$ is the tail dependence function of a random pair $(X, Y)$ if and only if 
\begin{enumerate}
	\item $\Lambda$ is homogeneous of order $1$, i.e., $\TDL{ax,ay} = a \TDL{x,y}$ for any $a > 0$
	\item $0 \leq \Lambda(x, y) \leq \min\rbraces{x, y}$ for all $x, y \geq 0$
	\item $\Lambda$ is concave.
\end{enumerate}
In particular, every tail dependence function is continuous.
\end{proposition}

\begin{remark} \label{rem:simplex}
In view of the homogeneity, we may, without losing any information, restrict any tail dependence function $\TDL{x,y}$ to the one-dimensional segment $\{ (x,y) \mid y=1-x \}$. 
The induced function is a Lipschitz continuous, concave function on the unit interval; abusing notation, we still call it the tail dependence function and denote it by $\TDL{s}$ where $s\in [0,1]$. 
In this setting, the tail dependence coefficient is given by
\begin{equation} \label{eq:tdc_restricted}
\lambda = 2 \Lambda(1/2) .
\end{equation}
\end{remark}

The following examples show different types of such tail dependence functions, defined on $[0,1]$. 
Note that some of these functions are symmetric with respect to the point $1/2$, which is the case, for instance, if the two random variables $X$ and $Y$ are exchangeable---see Definition~\ref{def:exchangeable}; other tail dependence functions, however, exhibit a marked asymmetric pattern. Note also that some of the
pairs of tail dependence functions shown in the next examples are ordered, while others are not.

\begin{example} \label{ex:td_functions}
\begin{enumerate}
\item \label{ex:td_functions_1} Let $(X, Y)$ follow a multivariate normal distribution with correlation $\rho \in [-1, 1]$.
Then $(X, Y)$ is tail independent, i.e. fulfils $\Lambda_{(X,Y)}(x,y) = 0$ for all $x, y \geq 0$, whenever $\rho < 1$ holds and tail dependent with $\Lambda_{(X,Y)}(x,y) = \min\rbraces{x, y}$ in case of $\rho = 1$. 
\item \label{ex:td_functions_2} Proposition~\ref{prop:tdf_properties} yields a straightforward construction method for arbitrary tail dependence functions. 
Two asymmetric tail dependence functions are depicted in Figure~\ref{fig:examples}.
\end{enumerate}
\end{example}

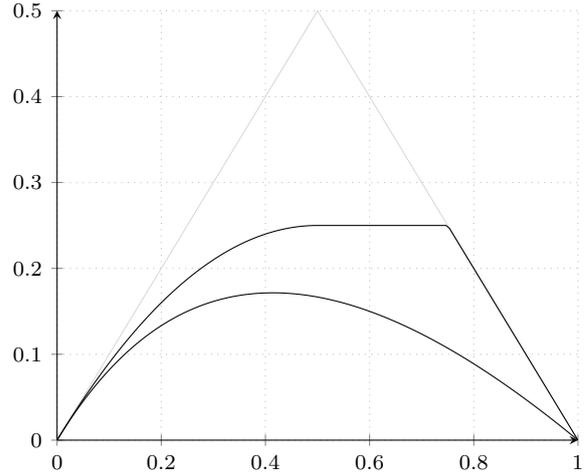
\begin{figure} 
\begin{tikzpicture}[ declare function={ 
    			tdf(\t) 	= min(\t, 1-\t); 
				tdf1(\t) 	= min(0.5*\t, 1-\t);
				tdf2(\t) 	= min(\t, 0.25*(1-\t));
				tdf3(\x,\y)	= \x*\y/(\x + \y);
				tdf13(\t)	= tdf3(\t, 0.5*(1-\t));
				tdf4(\t)	= (\t*(1-\t))*(\t < 0.5) + 0.25*(0.5 <= \t)*(\t < 0.75) + (1-\t)*(0.75 <= \t)*(\t < 1);
			}]
        \begin{axis}[mlineplot]
            \addplot[name path=f, domain=0:1, samples=101, color=gray!30]{tdf(\x)};
            \addplot[name path=f, domain=0:1, samples=151]{tdf13(\x)};
            \addplot[name path=f, domain=0:1, samples=151]{tdf4(\x)};
        \end{axis}
\end{tikzpicture}
\caption{Plot depicting asymmetric tail dependence functions.}
\label{fig:examples} 
\end{figure}

Returning to the tail dependence order, the following simple result shows that the tail dependence order implies the ordering of the tail dependence coefficient, but not vice versa.

\begin{proposition} \label{prop:tdc_monotone}
If $(X_1,Y_1) \tdo (X_2,Y_2)$ then $\lambda_1 \leq \lambda_2$; the converse is not true.
\end{proposition}

\begin{proof}
The first assertion is trivial in view of \eqref{eq:tdc_restricted}. 
The second assertion follows from Figure~\ref{fig:tdc_drawbacks}.
\end{proof}

\section{Measures of tail dependence} \label{section:measures}

While from a theoretical point of view the tail dependence function encodes all necessary information about the extremal behaviour, a practitioner will want to pinpoint more specific aspects of the extremal dependence into a single quantity. 
One possible application, for instance, could be the quantification of the overall extremal behaviour or certain diversification effects. 
This leads to the construction of \emph{measures of tail dependence}, i.e.\ functions $\mu$ that associate to each pair $(X,Y)$ a number $\mu(X,Y)$ which reflects a certain aspect of the underlying tail behaviour. 
In order to be able to compare different values of such a measure, we are looking for measures that are monotone with respect to the tail dependence order.

\begin{definition} \label{def:meas}
A measure of tail dependence is a function $\mu: (X,Y) \mapsto \mu(X,Y) \in [0,\infty)$, defined on all pairs $(X,Y)$ of random variables, satisfying the monotonicity condition $$ (X_1,Y_1) \tdo (X_2,Y_2) \Longrightarrow \mu(X_1,Y_1) \leq \mu(X_2,Y_2) . $$
 \end{definition}
 
It is important to note the fundamental role played by the underlying tail dependence order here. Of course, every measure $\mu$ can be used to define a preorder on the set of pairs $(X,Y)$ of random variables by setting
\begin{equation*}
    (X_1,Y_1) \leq_\mu (X_2,Y_2) :\Longleftrightarrow \mu(X_1,Y_1) \leq \mu(X_2,Y_2) .
\end{equation*}
This preorder $\leq_\mu$ is always reflexive, transitive and total but, in general, neither symmetric nor antisymmetric. 
Without the underlying tail dependence order, however, this construction may lead to grave inconsistencies as illustrated by Figure~\ref{fig:inconsistent} showing two tail dependence functions $\TDI{1}{s}$ and $\TDI{2}{s}$ which are not ordered. 
Evaluating the two measures of tail dependence from Theorem~\ref{thm:meas} given by $\mu(\Lambda) := \max_{s\in [0,1]} \TDL{s}$ and $\nu(\Lambda) := \cInt{\TDL{s}}{0}{1}{s}$, however, yields $\mu(\TDI{1}{s}) > \mu(\TDI{2}{s})$ and $\nu(\TDI{1}{s}) < \nu(\TDI{2}{s})$ so that we would end up with the inconsistent relations $$ (X_1,Y_1) >_{\mu} (X_2,Y_2) \text{ and } (X_1,Y_1) <_{\nu} (X_2,Y_2) . $$

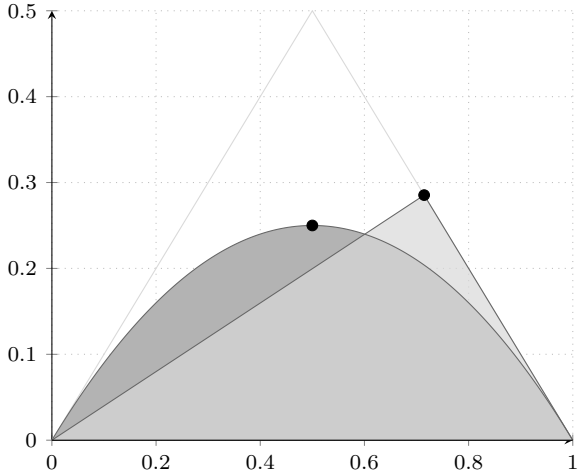
\begin{figure} 
\begin{tikzpicture}[ declare function={ 
    			tdf(\t) 	= min(\t, 1-\t); 
				tdf1(\t)	= \t*(1-\t);
				tdf2(\t) 	= min((1/3+0.066)*\t, 1-\t);
			}]
        \begin{axis}[mlineplot]
            \addplot[name path=f, domain=0:1, samples=151, color=gray!30]{tdf(\x)};
            \addplot[name path=g, domain=0:1, samples=151, color=black!60]{tdf1(\x)};
            \addplot[name path=h, domain=0:1, samples=151, color=black!60]{tdf2(\x)};
			\path[draw,name path = xAxis] (axis cs:0,0) -- (axis cs:1,0);
            \addplot[color=black, fill=gray!60] fill between [of=xAxis and g]; 
            \addplot[fill=gray!30, opacity=0.7] fill between [of=xAxis and h]; 
            \addplot[mark=*] coordinates {(1/2, {tdf1(0.5)})};
            \addplot[mark=*] coordinates {((1/(4/3 + 0.066)), {tdf2((1/(4/3 + 0.066)))})};
        \end{axis}
\end{tikzpicture}  
\caption{Plot depicting $\Lambda_1(s) = \min\curly{0.4s, 1-s}$ and $\Lambda_2(s) = s(1-s)$.
While the shaded area marks the average tail dependence $\nu$, the dot marks the maximal value $\mu$.}
\label{fig:inconsistent} 
\end{figure}

For the following, recall that $\TDL{s}$ denotes the (induced) tail dependence function on the unit interval $[0,1]$; see Remark~\ref{rem:simplex}.

\begin{theorem} \label{thm:meas}
Consider pairs $(X,Y)$ of random variables with tail dependence functions $\TDL{s}$. All the following quantities are measures of tail dependence:
\begin{enumerate}
\item the tail dependence cofficient $$ \lambda = 2\TDL{1/2} $$ and, more generally, any other value $\TDL{s_0}$ for some fixed $s_0 \in [0,1]$
\item the maximal tail dependence $$ \max_{s\in [0,1]} \TDL{s} $$
\item the average tail dependence $$ \cInt{\TDL{s}}{0}{1}{s} $$
\item any $L_p$-norm $$ \cbraces{\cInt{(\TDL{s})^p}{0}{1}{s}}^{1/p} $$ with $1\leq p < \infty$.
\end{enumerate}
\end{theorem}

\begin{proof}
All of the above quantities are well defined since $\TDL{s}$ is continuous, hence integrable, in view of Proposition~\ref{prop:tdf_properties}. The monotonicity condition is obviously satisfied by all quantities.
\end{proof}

In fact, our approach to measuring tail dependence is quite flexible, as measures of tail dependence can be constructed from any of the well-known measures of concordance (such as Spearman's $\rho$ or Kendall's $\tau$). For this, given a tail dependence function $\Lambda$, we consider the corresponding extreme-value copula defined by
\begin{equation*}
    C^{EV}_\Lambda(u, v) := \exp(\log{u} + \log{v} + \Lambda(-\log{u}, -\log{v})) ~;
\end{equation*}
see, e.g., \citet{Joe.2015}. Then, any measure of concordance applied to $C^{EV}_\Lambda$ generates a measure of tail dependence. In particular, we obtain the following result.

\begin{theorem}
Consider pairs $(X,Y)$ of random variables with tail dependence functions $\TDL{s}$.
Then Spearman's $\rho$ induces a tail dependence measure $\mu$ via
\begin{equation*}
    \mu(X, Y)   := \rho(C^{EV}_\Lambda) 
                = 12 \cInt{\frac{1}{(2 - \TDL{s})^2}}{0}{1}{s} - 3 ~.
\end{equation*} 
\end{theorem}

Note that in practical applications the choice of a measure of tail dependence will depend on the particular course of investigation---in some cases the maximal tail dependence may be appropriate while in others the average is more convenient. 
In any case, the evaluation measures $\TDL{s_0}$ for some point $s_0\in [0,1]$, in particular the classical tail dependence coefficient $\lambda$, are quite arbitrary choices. Only in very special cases does the tail dependence coefficient provide a reasonable measure of tail dependence, as we will see next.

\begin{definition} \label{def:exchangeable}
Two random variables $X$ and $Y$ are exchangeable if and only if $$ \Pb (X\leq x, Y\leq y) = \Pb (X\leq y, Y\leq x) $$ for all $x,y$; this is equivalent to saying that the joint distribution functions of $(X,Y)$ and $(Y,X)$ agree.
\end{definition}

\begin{proposition}
If $X$ and $Y$ are exchangeable then $$ \TDC = 2\TDL{1/2} = 2\max_{s\in [0,1]} \TDL{s} = 2 \| \TDL{s} \|_{\infty} . $$
\end{proposition}

\begin{proof}
$\TDL{s}$ is a continuous, concave function on the compact intervall $[0,1]$ and hence attains its maximal value on a closed, convex subset of $[0,1]$, i.e., on some compact subinterval $I\subset [0,1]$. Since $X$ and $Y$ are exchangeable $\TDL{s}$ is symmetric with respect to the point $1/2$, and we conclude that $1/2 \in I$.
\end{proof}

%

Finally, we point out that any functional combination $f(\mu,\nu)$ of measures of tail dependence with $f: [0,\infty) \times [0,\infty) \to [0,\infty)$ being increasing in each argument are also measures of tail dependence themselves; a simple example is any linear combination $a\mu+b\nu$ with $a,b\geq 0$. Such combinations can be useful in applications where one would like to consider different aspects of tail dependence simultaneously.

\begin{example}
\cite{Frahm.2006} introduced the so-called extremal dependence coefficient $$ \varepsilon_L := \frac{\TDC}{2 - \TDC} = \frac{\TDL{1/2}}{1 - \TDL{1/2}} = f(\TDL{1/2}) $$ with $f(\mu) := \frac{\mu}{1-\mu}$ which investigates the asymptotic dependence between the componentwise minima and maxima of a pair $(X,Y)$ of random variables. Since the function $f$ is increasing, the quantity $\varepsilon_L$ is a measure of tail dependence.
\end{example}

\section{Empirical analysis} \label{section:empirical}

In this section, we illustrate the practical relevance of analyzing the tail dependence function as well as several \textit{alternative} measures of tail dependence in a financial risk application. More precisely, we estimate the tail dependence function as well as the measures given in Theorem \ref{thm:meas} for a large set of bivariate equally-weighted portfolios consisting of the S\&P 500 stock index and one of the S\&P 500's constituent individual stocks.\footnote{By fixing the S\&P 500 as one of the two time series in our portfolios, we are able to analyze the systematic tail risk of the constituent individual stocks in the cross-section. However, as we focus on tail dependence and not correlation, our analysis differs from traditional analyses of stocks' beta factors.}

The financial market data we use are retrieved from \textit{EIKON Refinitive} and cover the period from Jan/9/2019 to Dec/31/2020. We select equities that are included in the S\&P 500 over our full sample period and for which we have complete data. In total, our panel data sample consists of 454 individual stocks and 2,500 trading days. For each day and time series in our sample, we then compute the log return. Summary statistics for our full sample are given in Table \ref{tab:1}. As one can see from Table \ref{tab:1}, the data exhibit the usual stylized facts of financial time series with mean daily log returns being close to zero, return distributions being negatively skewed, and fat tails.

Next, for each pair of the S\&P 500 and a constituent stock's log return time series, we estimate the corresponding tail dependence functions using the estimator proposed by \citet{SCHMIDT.2006} and rolling windows with a fixed length of 500 trading days. Consequently, we estimate $454\times 2,000$ tail dependence functions and then compute the three measures of tail dependence given in Theorem \ref{thm:meas}.

To start our discussion of the results, we first plot three selected estimated tail dependence functions from the end of our sample period (i.e., during the Covid-19 pandemic) in Figure \ref{fig:4}. As one can see from the three displayed stock pairs of Domino's Pizza, Agilent Technologies, and Celanese (each in combination with the S\&P 500), empirically, we can observe all types of symmetric and asymmetric tail dependence in the cross-section of stock pairs. To get a better picture of the time-variation in the measures of tail dependence, we plot the time evolution of the three measures $\left\|\Lambda\right\|_1$, $\lambda$, and $\left\|\Lambda\right\|_{\infty}$ for the three previously studied equity/index-pairs in Figure \ref{fig:5}. All three pairs exhibit a considerable amount of time-variation in all three measures of tail dependence. Moreover, one can see that even though the coefficient of tail dependence $\lambda$ is often close to the $\left\|\Lambda\right\|_{\infty}$-measure, we can observe periods of time (e.g., in 2017 for the S\&P500/Celanese pair, and at the end of 2017 for the S\&P500/Domino's Pizza pair) when $\lambda$ is considerably smaller than $\left\|\Lambda\right\|_{\infty}$.

Next, in Table \ref{tab:2} and Figure \ref{fig:6}, we show descriptive statistics as well as the time evolution of the $L^1$, the $L^{\infty}$, as well as the tail dependence coefficient for all bivariate pairs with the cross-sectional mean being shown as a black line and the values of the respective measures that lie between the cross-sectional 5\%- and 95\%-quantiles at each point in time being highlighted in gray. Not surprisingly for a well-diversified index, the stock/index-pairs show a large cross-sectional variation in the degree of tail dependence. Additionally, all three measures exhibit a high degree of time-variation during our sample period as well as sudden jumps.

Our analyses so far have shown that bivariate financial data time series exhibit a high degree of time variation in all three measures of tail dependence. Moreover, both symmetric and asymmetric tail dependence functions can be observed in our data. All three measures, however, are moving in lockstep during most of our sample period. To examine in more detail the question whether the $L^1$ and $L^{\infty}$ actually differ from the tail dependence coefficient in empirical settings, we compare possible and estimated values of the $L^{\infty}$ measure in Figure \ref{fig:7}. More precisely, we first estimate the tail dependence coefficient for the three bivariate stock/index pairs used in our previous analyses (Domino's Pizza, Agilent Technologies, and Celanese). We then calculate the range of possible values of the $L^{\infty}$ measure given the estimates of the tail dependence coefficient and highlight these as gray-shaded areas in Figure \ref{fig:7}. Finally, we estimate the $L^{\infty}$ measure and include them as a black line in the Figure. The estimated values of the $L^{\infty}$ measure are usually at the lower bound of the range of possible values given the estimated value of the tail dependence coefficient. However, we can observe for all three bivariate pairs time periods in which the estimated level of the $L^{\infty}$ measure is in the mid-range of its possible values. In addition to this, the range of possible values is far from being small: for a given level of the tail dependence coefficient, the range of possible $L^{\infty}$ values can have a size of up to $0.2$.

\section{Conclusion} \label{section:conclusion}

In this paper, we propose a new order on the set of pairs of bivariate random vectors based on their respective tail dependence functions. 
We show that simple measures of maximum or average tail dependence, or any $L^p$-norm are monotone with respect to our tail dependence order. 
We then highlight how different measures of tail dependence can capture different aspects of the tail dependence function's behaviour. 
Most notably, we stress the fact that the coefficient of tail dependence that is frequently used in empirical applications neglects most of the information on a random vector's tail dependence.

We further compare the different tail dependence measures in an empirical study using stock returns of the constituents of the S\&P 500. Our results show that all measures show significant cross-temporal and cross-sectional variation and do not necessarily move in parallel. Moreover, when fixing the coefficient of (lower) tail dependence, alternative measures of tail dependence can take on values in a large range showing the need to carefully choose the right measure of tail dependence for different empirical purposes. Our results thus emphasize the need to look beyond the classical coefficients of lower and upper tail dependence in empirical applications.



\begin{figure}[H]
\begin{subfigure}[c]{0.4\textwidth}
\includegraphics{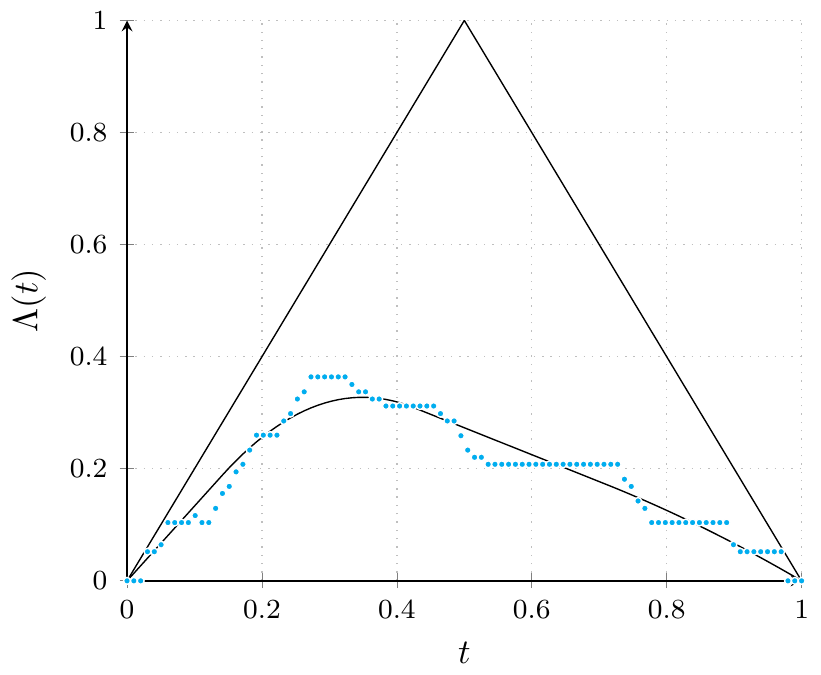}
\caption{S\&P500/Domino's Pizza Inc.}
\end{subfigure}
\begin{subfigure}[c]{0.4\textwidth}
\includegraphics{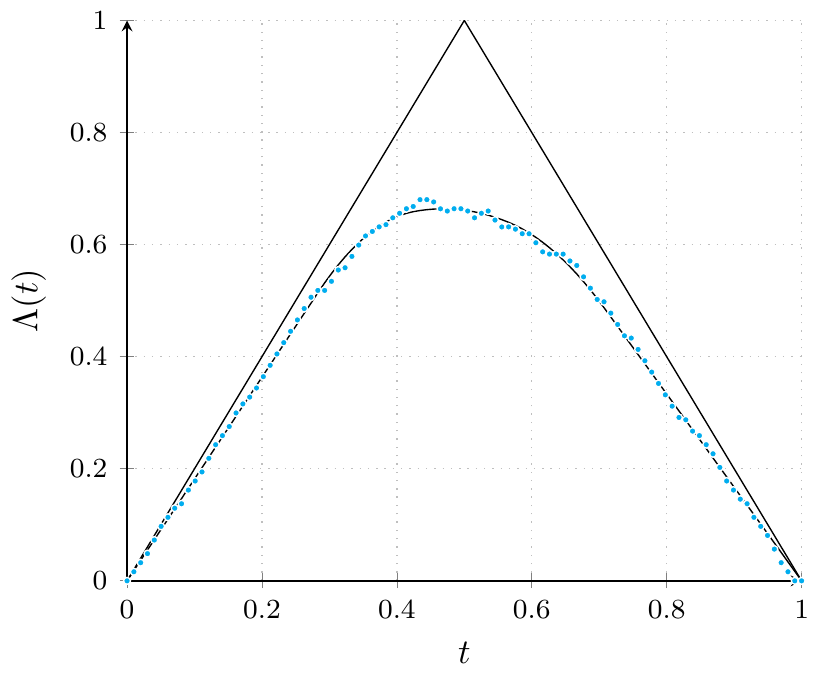}
\caption{S\&P500/Agilent Technologies Inc.}
\end{subfigure}
\begin{subfigure}[c]{0.4\textwidth}
\includegraphics{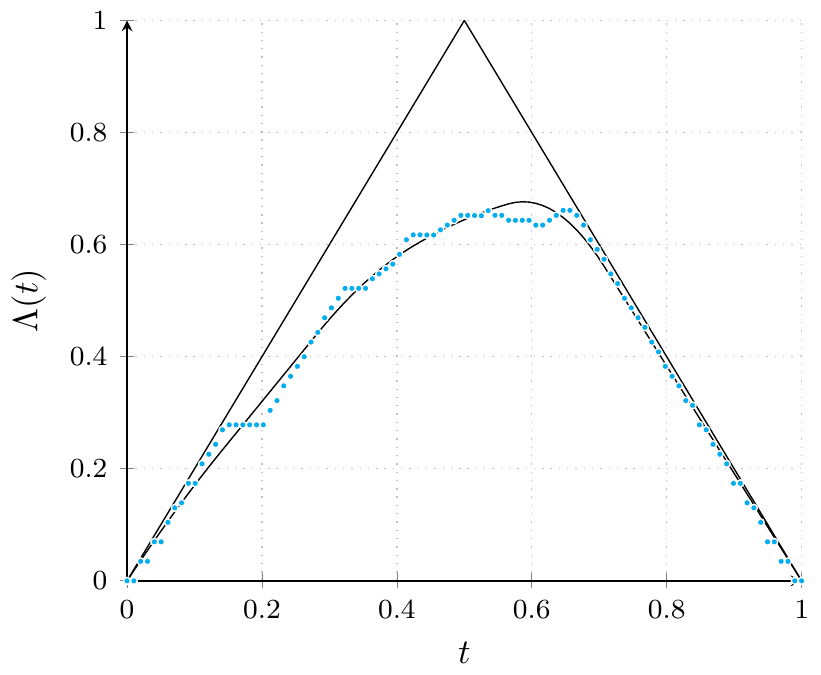}
\caption{S\&P500/Celanese Corp.}
\end{subfigure}
\caption{The figure shows three examples of empirical tail dependence functions estimated for the time period Jan/9/2019-Dec/31/2020 for three different bivariate pairs of the S\&P500 index with different individual stocks.}
\label{fig:4}%
\end{figure}


\begin{figure}[H]
\begin{subfigure}[c]{0.4\textwidth}
\includegraphics{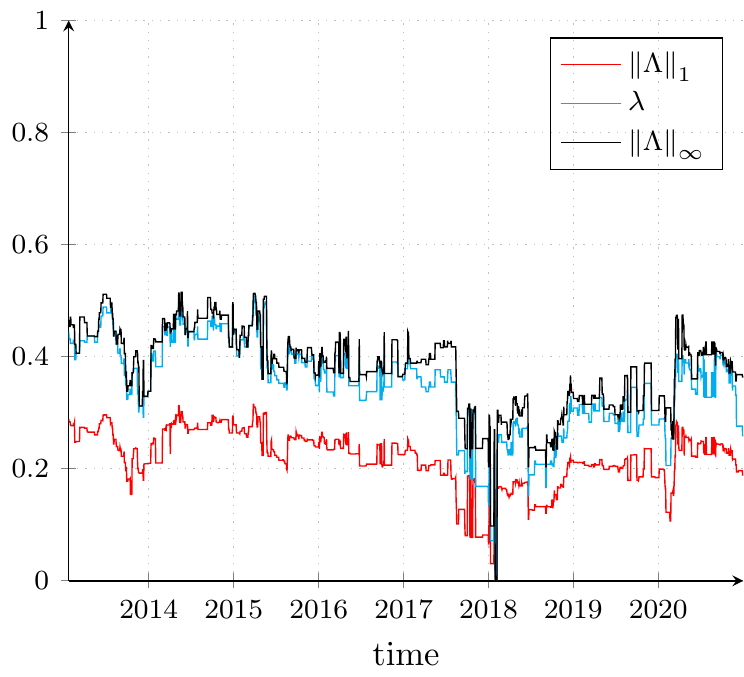}
\caption{S\&P500/Domino's Pizza Inc.}
\end{subfigure}
\begin{subfigure}[c]{0.4\textwidth}
\includegraphics{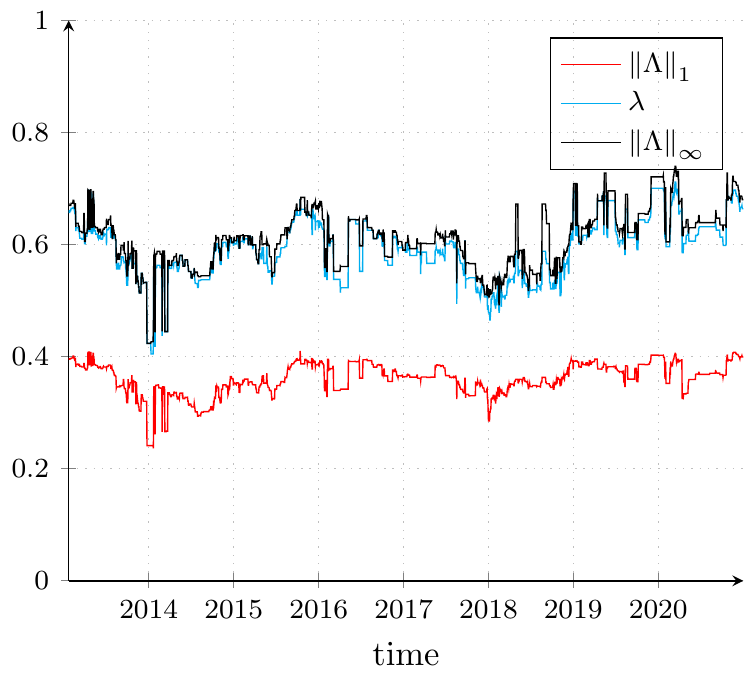}
\caption{S\&P500/Agilent Technologies Inc.}
\end{subfigure}
\begin{subfigure}[c]{0.4\textwidth}
\includegraphics{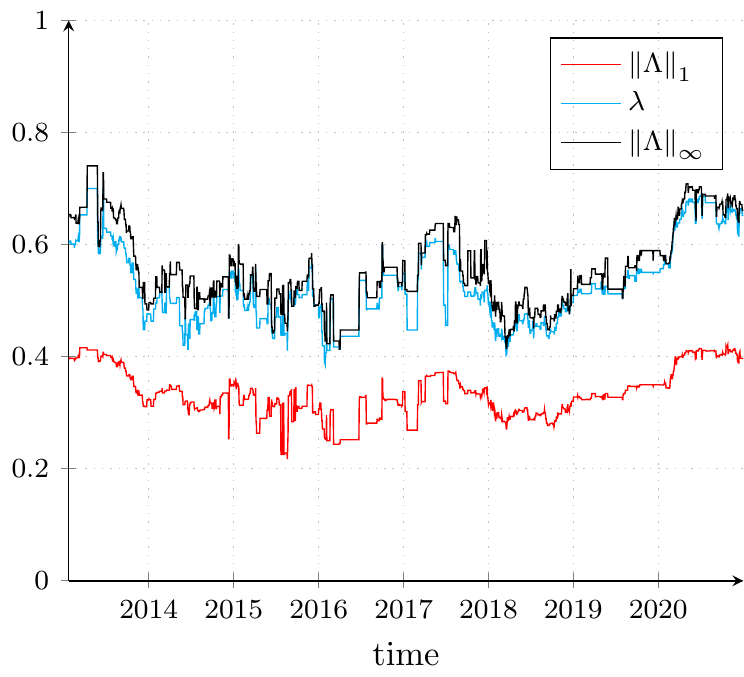}
\caption{S\&P500/Celanese Corp.}
\end{subfigure}
\caption{The figure shows the time evolution of the $L_1$ and $L_{\infty}$ measures, as well as the tail dependence coefficient for three selected bivariate pairs of the S\&P500 index with different individual stocks. The measures of tail dependence are estimated using 500-day rolling windows for the time period Jan/27/2011-Dec/31/2020.}
\label{fig:5}%
\end{figure}


\begin{figure}[H]
\begin{subfigure}[c]{0.4\textwidth}
\includegraphics{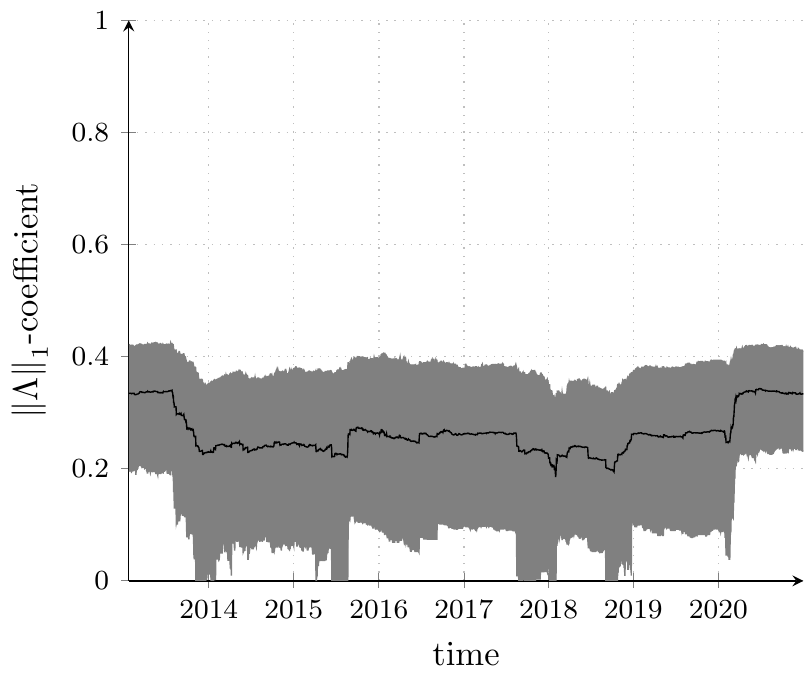}
\caption{$L_1$-norm}
\end{subfigure}
\begin{subfigure}[c]{0.4\textwidth}
\includegraphics{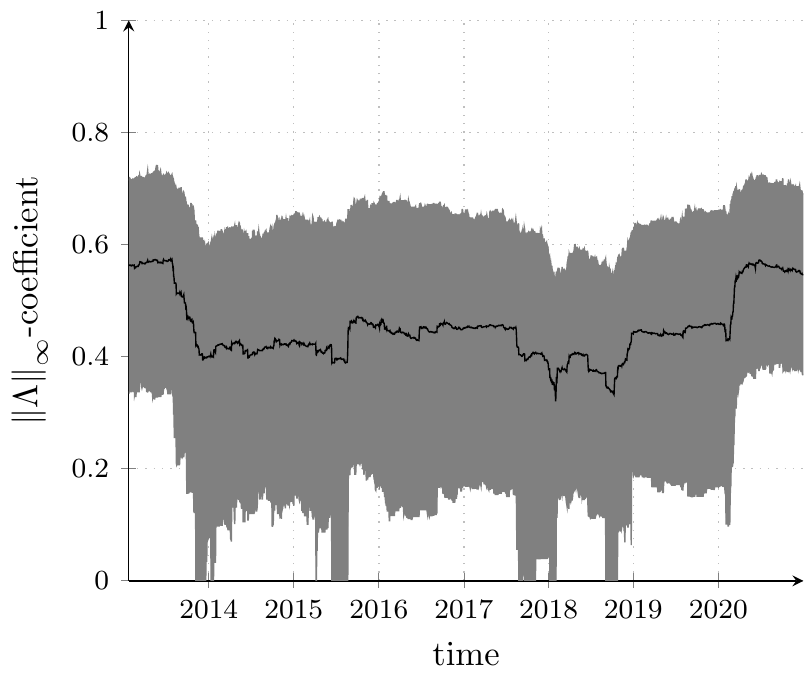}
\caption{$L_{\infty}$-norm}
\end{subfigure}
\begin{subfigure}[c]{0.4\textwidth}
\includegraphics{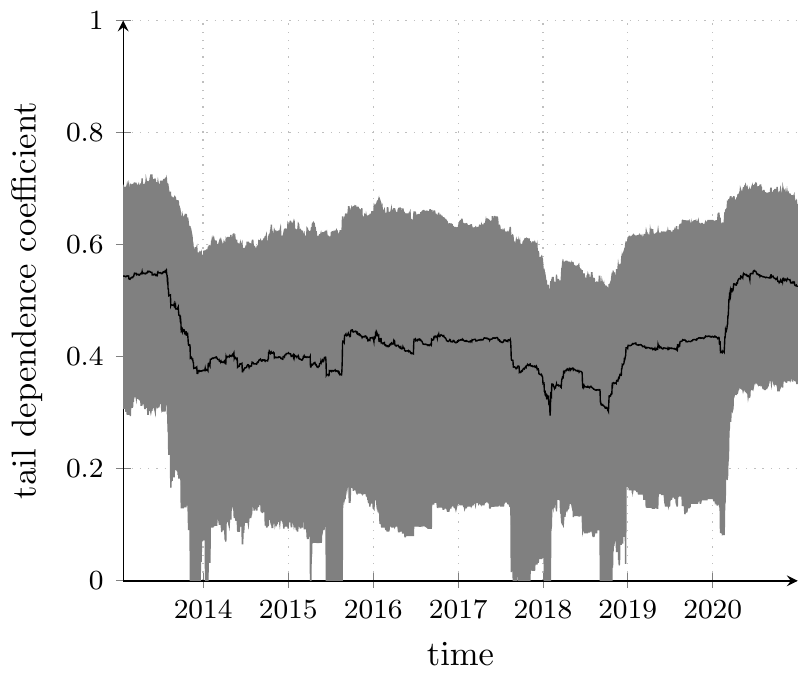}
\caption{Tail dependence coefficient $\lambda$}
\end{subfigure}
\caption{The figure shows the time evolution of the $L_1$ (a) and $L_{\infty}$ (b) measures, as well as the tail dependence coefficient (c) for all bivariate pairs between the S\&P500 index and the 454 constituent stocks for which we have full data in the time period Jan/27/2011-Dec/31/2020. All measures of tail dependence are estimated using 500-day rolling windows. The cross-sectional mean is shown as a black line while the values of the respective measures that lie between the cross-sectional 5\%- and 95\%-quantiles at each point in time are shown in gray.}
\label{fig:6}%
\end{figure}


\begin{figure}[H]
\begin{subfigure}[c]{0.4\textwidth}
\includegraphics{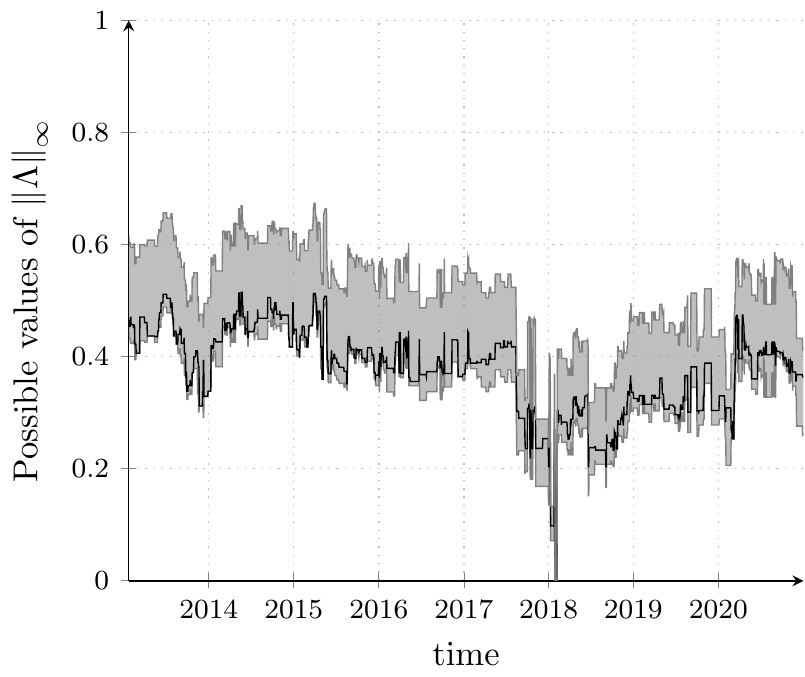}
\caption{S\&P500/Domino's Pizza Inc.}
\end{subfigure}
\begin{subfigure}[c]{0.4\textwidth}
\includegraphics{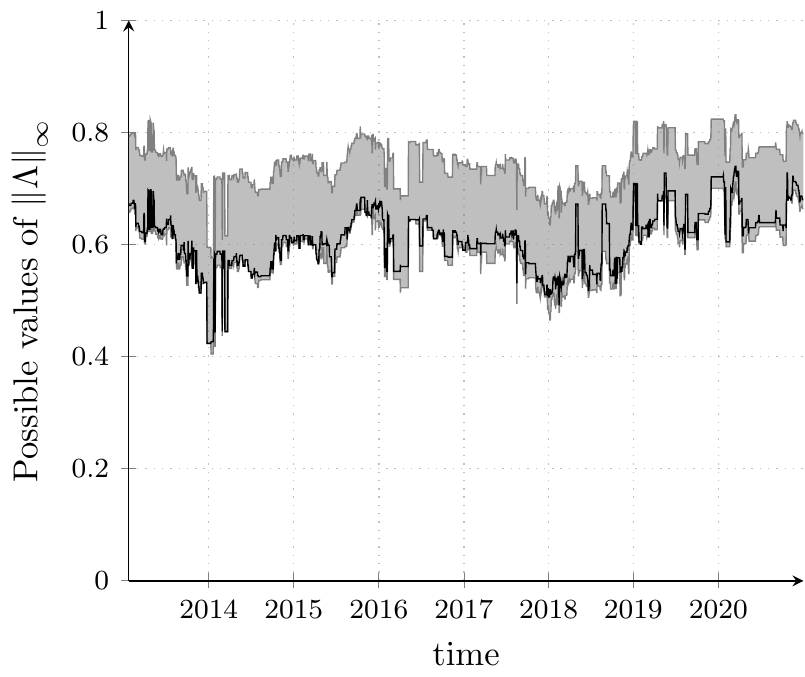}
\caption{S\&P500/Agilent Technologies Inc.}
\end{subfigure}
\begin{subfigure}[c]{0.4\textwidth}
\includegraphics{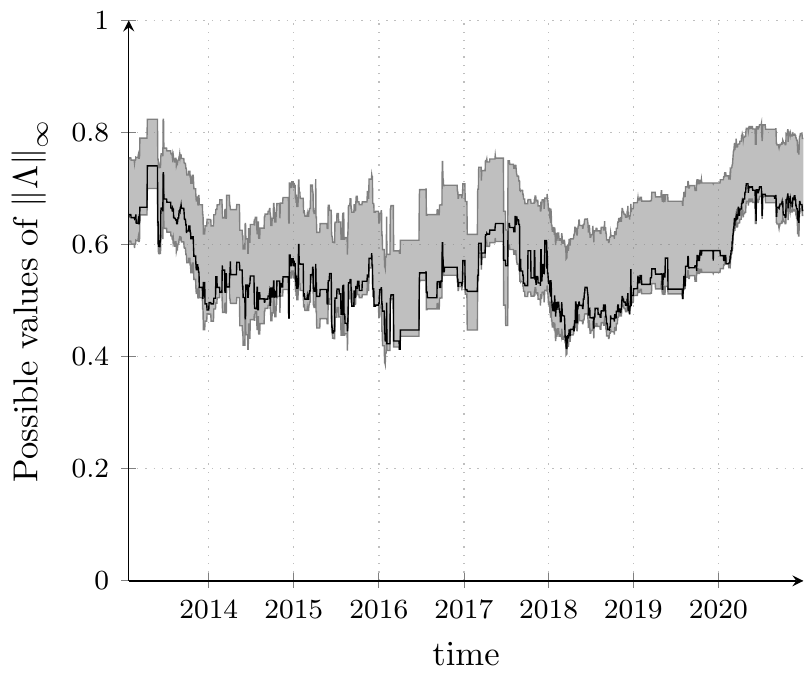}
\caption{S\&P500/Celanese Corp.}
\end{subfigure}
\caption{The figure shows the possible (gray area) and estimated (black line) values of the $L_{\infty}$ measure given the estimated level of the tail dependence coefficient for three selected bivariate pairs of the S\&P500 index with different individual stocks. The measures of tail dependence are estimated using 500-day rolling windows for the time period Jan/27/2011-Dec/31/2020.}
\label{fig:7}%
\end{figure}


\begin{table*}[]
\begin{tabular}{lcccccccc}
\hline\hline
                       & \textit{S\&P500} & \multicolumn{1}{c}{\textit{}} & \multicolumn{6}{c}{\textit{Individual Stocks (cross-section)}}                                                                                                                                                                                           \\ \cline{4-9} 
\textit{Statistic}                       &                    & \textit{}                     & \multicolumn{1}{c}{\textit{5\%-quantile}} & \multicolumn{1}{c}{\textit{10\%-quantile}} & \multicolumn{1}{c}{\textit{Mean}} & \multicolumn{1}{c}{\textit{Median}} & \multicolumn{1}{c}{\textit{90\%-quantile}} & \multicolumn{1}{c}{\textit{95\%-quantile}} \\ \hline
\textit{Mean}          & 0.0004                               &                               & -0.0001                                   & 0.0000                                     & 0.0004                            & 0.0004                              & 0.0009                                     & 0.0010                                     \\
\textit{Median}        & 0.0006                               &                               & 0.0000                                    & 0.0003                                     & 0.0007                            & 0.0007                              & 0.0011                                     & 0.0013                                     \\
\textit{St. dev.}      & 0.0110                               &                               & 0.0126                                    & 0.0134                                     & 0.0188                            & 0.0179                              & 0.0258                                     & 0.0278                                     \\
\textit{Minimum}       & -0.1277                              &                               & -0.3463                                   & -0.2853                                    & -0.1907                           & -0.1722                             & -0.1125                                    & -0.1025                                    \\
\textit{Maximum}       & 0.0897                               &                               & 0.0960                                    & 0.1035                                     & 0.1612                            & 0.1480                              & 0.2437                                     & 0.2719                                     \\
\textit{5\%-quantile}  & -0.0164                              &                               & -0.0389                                   & -0.0356                                    & -0.0267                           & -0.0257                             & -0.0190                                    & -0.0177                                    \\
\textit{95\%-quantile} & 0.0147                               &                               & 0.0172                                    & 0.0186                                     & 0.0262                            & 0.0250                              & 0.0351                                     & 0.0395                                     \\ \hline\hline
                       &                                      &                               &                                           &                                            &                                   &                                     &                                            &                                           
\end{tabular}
\caption{The table presents summary statistics for the S\&P 500 index as well as the 454 constituent equities that are employed in the empirical study. The sample period runs from Jan/27/2011 to Dec/31/2020 and includes 2,520 trading days.}
\label{tab:1}%
\end{table*}

\begin{table*}[]
\begin{tabular}{clccccccc}
\hline\hline
\multicolumn{1}{l}{}                   & \textit{}              & \textit{} & \multicolumn{6}{c}{\textit{S\&P500/stock pairs   (cross-section)}}                                                                                                                                                                                       \\ \cline{4-9} 
\multicolumn{1}{l}{\textit{Tail dep. measure}}                   & \textit{Statistic}     & \textit{} & \multicolumn{1}{c}{\textit{5\%-quantile}} & \multicolumn{1}{c}{\textit{10\%-quantile}} & \multicolumn{1}{c}{\textit{Mean}} & \multicolumn{1}{c}{\textit{Median}} & \multicolumn{1}{c}{\textit{90\%-quantile}} & \multicolumn{1}{c}{\textit{95\%-quantile}} \\ \hline
\multirow{7}{*}{\textit{$L_1$}}        & \textit{Mean}          &           & 0.0744                                    & 0.0818                                     & 0.1318                            & 0.1344                              & 0.1743                                     & 0.1805                                     \\
                                       & \textit{Median}        &           & 0.0669                                    & 0.0800                                     & 0.1341                            & 0.1395                              & 0.1781                                     & 0.1839                                     \\
                                       & \textit{St. dev.}      &           & 0.0175                                    & 0.0202                                     & 0.0329                            & 0.0313                              & 0.0474                                     & 0.0541                                     \\
                                       & \textit{Minimum}       &           & 0.0000                                    & 0.0000                                     & 0.0428                            & 0.0365                              & 0.1060                                     & 0.1196                                     \\
                                       & \textit{Maximum}       &           & 0.1574                                    & 0.1649                                     & 0.1934                            & 0.1973                              & 0.2164                                     & 0.2200                                     \\
                                       & \textit{5\%-quantile}  &           & 0.0000                                    & 0.0000                                     & 0.0747                            & 0.0771                              & 0.1351                                     & 0.1449                                     \\
                                       & \textit{95\%-quantile} &           & 0.1431                                    & 0.1526                                     & 0.1832                            & 0.1879                              & 0.2087                                     & 0.2127                                     \\ \hline
\multirow{7}{*}{\textit{$L_{\infty}$}} & \textit{Mean}          &           & 0.1273                                    & 0.1414                                     & 0.2256                            & 0.2317                              & 0.2952                                     & 0.3052                                     \\
                                       & \textit{Median}        &           & 0.1211                                    & 0.1416                                     & 0.2290                            & 0.2359                              & 0.2993                                     & 0.3111                                     \\
                                       & \textit{St. dev.}      &           & 0.0316                                    & 0.0338                                     & 0.0531                            & 0.0491                              & 0.0788                                     & 0.0892                                     \\
                                       & \textit{Minimum}       &           & 0.0000                                    & 0.0000                                     & 0.0827                            & 0.0811                              & 0.1926                                     & 0.2092                                     \\
                                       & \textit{Maximum}       &           & 0.2657                                    & 0.2768                                     & 0.3325                            & 0.3374                              & 0.3783                                     & 0.3907                                     \\
                                       & \textit{5\%-quantile}  &           & 0.0000                                    & 0.0000                                     & 0.1342                            & 0.1449                              & 0.2301                                     & 0.2432                                     \\
                                       & \textit{95\%-quantile} &           & 0.2365                                    & 0.2512                                     & 0.3084                            & 0.3160                              & 0.3567                                     & 0.3668                                     \\ \hline
\multirow{7}{*}{\textit{TDC}}          & \textit{Mean}          &           & 0.1154                                    & 0.1267                                     & 0.2139                            & 0.2191                              & 0.2852                                     & 0.2953                                     \\
                                       & \textit{Median}        &           & 0.1025                                    & 0.1300                                     & 0.2170                            & 0.2238                              & 0.2892                                     & 0.3008                                     \\
                                       & \textit{St. dev.}      &           & 0.0327                                    & 0.0357                                     & 0.0539                            & 0.0501                              & 0.0767                                     & 0.0859                                     \\
                                       & \textit{Minimum}       &           & 0.0000                                    & 0.0000                                     & 0.0743                            & 0.0705                              & 0.1787                                     & 0.1929                                     \\
                                       & \textit{Maximum}       &           & 0.2523                                    & 0.2667                                     & 0.3235                            & 0.3299                              & 0.3719                                     & 0.3825                                     \\
                                       & \textit{5\%-quantile}  &           & 0.0000                                    & 0.0000                                     & 0.1224                            & 0.1273                              & 0.2175                                     & 0.2336                                     \\
                                       & \textit{95\%-quantile} &           & 0.2263                                    & 0.2395                                     & 0.2991                            & 0.3061                              & 0.3499                                     & 0.3584                                     \\ \hline\hline
\end{tabular}
\caption{Summary statistics: Tail dependence measures.}
\label{tab:2}%
\end{table*}


\bibliography{reference}
\end{document}